\documentclass[submission,copyright,creativecommons]{eptcs}
\providecommand{\event}{WWV 2015}

\usepackage{% a4wide,
  afterpage,%authblk,
  amsmath,amssymb,amsthm,
  pdfsync,
  url,xspace,doi
} 
  
\newtheorem{theorem}{Theorem}[section]
\newtheorem{corollary}[theorem]{Corollary}
\newtheorem{lemma}[theorem]{Lemma}

\newtheorem{example}[theorem]{Example} 
\newtheorem{definition}[theorem]{Definition}   
\newtheorem{remark}[theorem]{Remark}

\usepackage{pic}
\usepackage{types}
\usepackage{shorthand}
\usepackage{mathpartir}
\usepackage[inference]{semantic}
\usepackage{color}  % for comments 

\newcommand{\inn}[1]{\ensuremath{\textsf{in}(#1)}}
\newcommand{\outt}[1]{\ensuremath{\textsf{out}(#1)}}
\newcommand{\snc}[1]{\ensuremath{\textsf{sync}(#1)}}
\newcommand{\wt}[1]{\ensuremath{\textsf{wait}(#1)}}
\newcommand{\lf}[1]{\ensuremath{\textsf{lfree}(#1)}}

\newcommand{\cinn}[1]{\ensuremath{\textsf{cin}(#1)}}
\newcommand{\coutt}[1]{\ensuremath{\textsf{cout}(#1)}}
\newcommand{\complete}[1]{\ensuremath{\textsf{cmp}(#1)}}
\newcommand{\tcomplete}[1]{\ensuremath{\textsf{tcmp}(#1)}}

\newcommand{\dl}[1]{\ensuremath{\textsf{dlock}(#1)}}

\newcommand{\slfdl}[1]{\ensuremath{\textsf{sl}(#1)}}%MARCO: to use the same initials
\newcommand{\pslfl}[1]{\ensuremath{\textsf{psl}(#1)}}

\newcommand{\LF}{\textsf{LF}}
\newcommand{\SL}{\textsf{PSL}}

\newcommand{\CMP}{\textsf{CMP}}

\newcounter{ncomm}

\newcommand{\ignore}[1]{}
\newcommand{\oldProof}[1]{}

\title{Unlocking Blocked Communicating Processes}

\author{
Adrian Francalanza
\institute{CS, ICT,  University of Malta}
\and Marco Giunti
\institute{RELEASE,  DI, Universidade da Beira Interior} 
\institute{NOVA LINCS, DI-FCT, Universidade NOVA de Lisboa}
\and Ant\'onio Ravara 
\institute{NOVA LINCS,
  DI-FCT, Universidade NOVA de Lisboa}
}

\def\titlerunning{Unlocking Blocked Communicating Processes}
\def\authorrunning{A.~Francalanza, M.~Giunti and A.~Ravara}

\begin{document}

\maketitle

\begin{abstract}
  We study the problem of disentangling locked processes via %through
  code refactoring. We identify and characterise a class of
  processes that is not lock-free; then we
%  refactoring so as to obtain lock-free processes. We identify the
%  class of processes that our proposed analysis addresses, and give an
% alternative characterisation that is easier to work with. We
  formalise an algorithm that statically detects potential locks and
  propose refactoring procedures that disentangle detected locks.  Our
  development is cast within a simple setting of a finite linear \ccs
  variant --- although it suffices to illustrate the main concepts, we
  also discuss how our work extends to other language
  extensions. % such as value-passing.
\end{abstract}

%%% Local Variables:
%%% mode: latex
%%% TeX-master: "position-paper"
%%% End:

\section{Introduction} % \af{0.75pg}
\label{sec:introduction}
\paragraph{The scenario.}%
% Concurrency is these days pervasive in computational systems, not only
% due to the multi-core architectures or cloud infrastructures, but also
% due to the fact that devices run several apps at the same time, using
% common resources.
%
Concurrent programming is nowadays pervasive to most computational
systems and present in most software development processes.  In
particular, concurrent programming is prevalent in cloud platforms and
web-services, \ie inherently distributed systems that rely heavily on
message-based communication protocols. Unfortunately, this style of
programming is notoriously difficult and error-prone: concurrency bugs
appear frequently and have a substantial impact, as several recent
reports
show~\cite{DBLP:conf/dsn/FonsecaLSR10,DBLP:conf/asplos/LuPSZ08}.
Concurrency errors are hard to detect because not every execution
interleaving exhibits them, and this is further compounded by the
large number of possible execution scenarios. Automatic techniques and
tools are thus needed to analyse and ensure correct concurrent code.

One common form of bugs is that of (dead)locks~\cite{coffman1971}:
they arise when a computational entity holds exclusive access to a
resource without releasing it, while other entities % circularly
wait to access that resource. In this work we characterise them in a
very simple model of concurrent computation, show how to statically
detect them, and in some cases, even show how to automatically solve
some of the (dead)locks.

\paragraph{Static analysis to the rescue.}%
Concurrency theory is a well-established discipline, providing
mathematical models of concurrent systems at various degrees of
expressiveness, (logical) languages to specify properties of such systems,
suites of verification techniques of both safety and
liveness properties, as well as tools to (automatically) analyse if
some property holds for a given specification.

We are interested in models centered around communication primitives
and synchronisation mechanisms, as these are the key characteristics
of a concurrent system. In particular, we are concerned with the
static verification of properties for these models, not only because
the approach analyses source code, but also because it is used
pre-deployment, in an automatic way. The models are useful to specify
and verify communication intensive systems and protocol
implementations; the static analysis is a light verification
technique % in the sense
that demands less from the user, as (s)he does not have to be an
expert in logic.

Concretely, herein we use the Calculus of Communicating Systems
(\ccs)~\cite{DBLP:books/sp/Milner80} and define a static analysis and
refactoring algorithm that is not only fully automatic, but also
working on ``pure'' source code, without further annotations or even types.

\paragraph{Behavioural types.} This field of study has gained momentum
recently by providing statically more than the usual safety
properties: (dead)lock-freedom or even progress can be statically
established by relying on a broad spectrum of analysis
techniques~\cite{CairesVieira10,CarbDM14,CarboneMontesi13,DBLP:conf/coordination/CoppoDPY13,GiachinoKL:deadlocksUnboundProcessNets,GiuntiR13,Kobayashi:new-type-system,Pad14,vieira.vasconcelos:typing-progress-communication-centred-systems}.
Despite their utility, such static detection techniques inevitably
approximate the solution (conservatively) --- since they are
``deciding % and
undecidable properties'' --- and reject %otherwise
lock-free programs (give false positives). More importantly, however,
these techniques simply reject programs, without providing help as to
where the problem might be, or providing insights on how to solve the
problem detected.
% Static deadlock detection is conservative -- since the algorithms need
% to ``decide and undecidable property'', they over-approximate and give
% false positives,\ie identify as erroneous programs that are
% correct. Nonetheless, they are very helpful as they analyse source
% code at compile time and give correctness guarantees (``accepted
% programs will not go wrong or get stuck''). However, the tools usually
% do not provide insight on how to solve a problem when they find
% it. They simply reject the code without giving any help, namely by
% identifying where the problem might be.

\paragraph{Methodology.} Following the approach of Giunti and
Ravara~\cite{GiuntiR13}, in this paper we propose that such
techniques % should
go a step further, and provide suggestions on how to fix a detected
bug, showing possible patches.
In particular, the work in \cite{GiuntiR13} focussed on resolving
self-holding (dead)locks, \ie when a thread holds the resources it
wants to use itself. In order to detect such errors, a local analysis
within one thread of computation, such as those discussed in
\cite{GiuntiR13}, sufficed.

By contrast, in this paper we investigate methods for resolving
circular-wait (dead)locks, \ie more general instances where concurrent
entities block one another by holding access to a subset of the
commonly requested resources.  Detecting such (dead)locks requires
analyses that spread across parallel threads of computation, and one
of the main challenges is to devise static techniques that are
compositional (thus scalable) \wrt the independent computing entities.
For this expository paper, we do not consider the full language of
\cite{GiuntiR13}.  Instead we distill minimal features of the language
--- namely synchronisation, prefixes, and parallel composition ---
that permit us to focus the core challenges to be addressed.  However,
the ultimate aim of the work is still to address circular-wait
(dead)locks in the full language of \cite{GiuntiR13}; in the
conclusion, we outline some of the additional issues that arise in the
full language setting.

% Therefore, static automatic code refactoring to eliminate (deadlock)
% bugs is a contribution that may have an important impact: instead of
% simply telling the programmer ``congratulations, you're deadlocked!'',
% such tools would provide suggestions on how to fix the bug, showing
% possible patches.

%\paragraph{The paper is structured as follows.}%
\paragraph{Contributions.}%
\S~\ref{sec:language} briefly introduces %our expository language.
a concise, yet sufficiently expressive, process language to rigorously
define (dead)locks.
In \S~\ref{sec:lock-freedom} we formalise the class of non lock-free
processes targeted by our work and give an alternative
characterisation for this class.
We present an algorithm for statically detecting processes in this
class in \S~\ref{sec:stat-detect-potent},
and in \S~\ref{sec:disent-potent-self} we describe disentangling
procedures for the detected processes. 
\S~\ref{sec:conclusion} concludes.  
We note that whereas \S~\ref{sec:lock-freedom} presents formal
results, \S~\ref{sec:stat-detect-potent} and
\S~\ref{sec:disent-potent-self} deal with ongoing work. In particular,
they present our general approach by formalising potential algorithms
for static analysis and resolution, and outlining the properties that
these algorithms are expected to satisfy.

%%% Local Variables:
%%% mode: latex
%%% TeX-master: "position-paper"
%%% End:

\section{Language} % \af{0.75pg}
\label{sec:language}
\begin{figure}[t]
  \centering
  \begin{align*}
     P,Q,R\in \GProc\bnfdef  &&
        \inert&& (\text{inert})  &
      \bnfsepp \prf{\act}P&& (\text{prefix}) & 
      \bnfsepp P\paral Q&& (\text{composition}) 
  \end{align*}
  \begin{align*}
    \ectxN & \bnfdef [-] \bnfsep P \paral \ectxN  \bnfsep   \ectxN \paral P   &&
                                                                               \quad\text{(Evaluation Contexts)} \\% &\quad\;
    \ctxN & \bnfdef [-] \bnfsep P \paral \ctxN  \bnfsep   \ctxN \paral P \bnfsep \prf{\act}\ctxN  &&
                                                                                \quad \text{(Process Contexts)}\\
  \end{align*}
  \begin{mathpar}
    \rtit{sNil}\quad P \paral \inert \;\steq\; P \and
    \rtit{sCom}\quad P \paral Q \;\steq\; Q \paral P   \and 
    \rtit{sAss} \quad P\paral(Q \paral R) \;\steq\; (P\paral Q) \paral R   \\
  \end{mathpar}
  \begin{mathpar}
    \inference[\rtit{Com}]{}{\prf{a}P\paral\prf{\bar{a}}Q \reduc P\paral Q} % \and
    \quad
    \inference[\rtit{Ctx}]{P \reduc P'}{\ectx{P} \reduc \ectx{P'}} % \and
    \quad
    \inference[\rtit{Str}]{P\steq P'\reduc Q'\steq Q}{P\reduc Q} 
  \end{mathpar}
  \caption{The language (finite \ccs): syntax and operational semantics}
  \label{fig:language}
% \hrulefill
\end{figure}

We consider a very basic language. % to focus on the usefulness of the
% concepts and techniques we propose.
% \smallskip 
Assume a countable set $\Names$ of \emph{names}, ranged over by
$a,b,\ldots,$ % possibly primed or indexed,
and a disjoint countable set
$\overline\Names$ of \emph{co-names}, % disjoint from the set of names and
such that for every $a \in \Names$ there is a $\overline a \in
\overline\Names$; % (\ie $\Names$ is equipped with a bijection).
 the co-action operation is idempotent \ie % observes the property
$\bar{\bar{a}} = a$,
% \smallskip 
% Assume as well a set of
and let 
$\act,\actt \in \bigl(\Names \cup \overline{\Names}\bigr)$ denote actions.  % The
% language is thus the finite \ccs (without recursion or name
% scoping).

The grammar in \figref{fig:language} defines the syntax of the
language, a process algebra containing only prefixing and parallel
composition, together with action synchronisations akin to \ccs
\cite{DBLP:books/sp/Milner80}.
%
%We write $\ctx{Q}$ (\resp $\ectx{Q}$) for
Let $\ctx{Q}$ (\resp $\ectx{Q}$) be the process obtained by
substituting the hole $[-]$ occurring in the context $\ctxN$ (\resp
$\ectxN$) with $Q$.

The semantics is standard, relying on a structural
equivalence relation~$\equiv$ (the smallest congruence including the
relation inductively generated by the rules below the grammar) and on
a reduction relation~$\reduc$, inductively generated by the
rules % at the bottom
of \figref{fig:language}.
Let $\reducAst$ denote the reflexive and transitive closure of~$\rightarrow$.
%
% We indicate with $\Names(P)$ the subset of $\Names$ induced by the rule
% $\Names(a.P) = \{a\}\cup\Names(P)$: the remaining cases are homomorphic. 
% \smallskip

Finally, assume henceforth a type system enforcing a linear use of
names, \tproc{\env}{P}, along the lines of the work of Kobayashi
\cite{kobayashi:type-systems}.
%typeable % \footnote{Besides enforcing a linear discipline, the type system verifies that the read/write capabilities on names contained in the type environment~$\Gamma$ are used accordingly.}
In well-typed processes, % $\tproc{}{P} \deftxt \exists \env. \tproc{\env}{P}$,
no name appears more than once with a given capability (input or
output), \ie a name occurs at most twice in a process, or none at all.
The set $\Proc$ is the subset of $\GProc$ induced by the typing
system~$\vdash$.

%%% Local Variables:
%%% mode: latex
%%% TeX-master: "position-paper"
%%% End:

\section{Lock Freedom}% \af{1.5pg}
\label{sec:lock-freedom}
% The aim of this section is to provide a characterization of processes that are not lock-free and that cannot be 
% unlocked  by the context.
%
Our point of departure is \emph{lock-freedom}, as defined and studied
by Kobayashi and by Padovani~\cite{kobayashi:type-systems,Pad14}.

\begin{definition}[%Padovani 
Synchronisation predicates~\cite{Pad14}]\label{def:pad-basic-pred}
  \begin{align*}
    \inn{a,P} &\;\deftxt\; \exists P',P'' \cdot P\steq P'\paralS \prf{a}P''  & \outt{a,P} &\;\deftxt\; \exists P',P'' \cdot P\steq P'\paralS\prf{\bar{a}}P''\\
    \snc{a,P} & \;\deftxt\; \inn{a,P} \;\text{ and }\;\outt{a,P} & \wt{a,P} & \;\deftxt\; \inn{a,P} \;\text{ exor }\;\outt{a,P}    
  \end{align*}
 \end{definition}
  
\begin{definition}[Lock-Free~\cite{Pad14}]\label{def:lf} We define % the subclass
  $\LF \deftxt \sset{P\in \GProc \;|\; \lf{P}}$ where:
    \begin{align*}
    \lf{P} & \;\deftxt\; P \reduc^\ast Q \text{ and } \wt{a,Q} \quad \text{implies}\quad \exists R \cdot Q \reduc^\ast R  \text{ and } \snc{a,R}
  \end{align*}
\end{definition}

%\noindent We note that $\inert$ is lock-free.
%We should talk about the class $\Proc\setminus\LF$.  
% In the finite case, they are processes that reach a deadlock (see Def.~\ref{def:dl}) whereas in the general case, these may be more difficult to characterise. 

% The question we face in this section is: which is the shape and the behaviour of processes that are not 
% lock-free?  Intuitively, they are processes that \emph{never} provide the co-action for some action.
Following \defref{def:lf}, locked processes, $\Proc\setminus\LF$, are
those that \emph{never} provide the \resp co-action for some waiting
action. In the setting of \secref{sec:language}, this could be due to
either of two cases: $(i)$ the co-action is not present in the
process; $(ii)$ the co-action is present, but stuck underneath a
blocked prefix.
% \begin{enumerate}
% \item The co-action is not present in the process.
% \item The co-action is present, but stuck underneath a blocked prefix.
% \end{enumerate}
Whereas in the case of $(i)$, the context may unlock (\ie catalyse
\cite{CarbDM14}) the process by providing the necessary co-action, in
the case of $(ii)$ \emph{no} context can do so \emph{without violating
  the linear discipline} of the process.  Our work targets the
unblocking of this second class of locked processes, specifically by
\emph{refactoring} the prefixing of the existing process. To this aim,
we introduce the notion of a \emph{complete} process.

% we note that in case 2 the context may help only if 
% names do not follow a linear discipline: that is, when names are linear the context is irrelevant.
% Our main interest is to study this class of processes, that is processes of 2 following a linear discipline; in \S~\ref{sec:stat-detect-potent} and \S~
% \ref{sec:disent-potent-self} we will devise mechanisms to statically detect elements of the (sub) class and 
% program transformations that map these processes into lock-free processes.

\begin{definition}[Complete Processes]\label{def:complete} % We define
  % the class
  $ \Proc\supseteq\CMP  \deftxt \sset{P \;|\; \complete{P}}$    where:
 \begin{align*} 
      & \quad\complete{P}  \;\deftxt\; \forall a \cdot \bigl(\cinn{a,P} \text{ iff } \coutt{a,P} \bigr) \qquad \text{ and }\\
    \cinn{a,P} & \;\deftxt\; \exists \ctx{-},Q \cdot P \steq \ctx{Q} \text{ and } \inn{a,Q}      
   \quad\qquad\coutt{a,P}  \;\deftxt\; \exists \ctx{-},Q \cdot P \steq \ctx{Q} \text{ and } \outt{a,Q} 
    \end{align*}
 \end{definition}

 \begin{remark} In contrast to \inn{a,P} and \outt{a,P} of
   \defref{def:pad-basic-pred}, the predicates \cinn{a,P} and
   \coutt{a,P} of \defref{def:complete} consider actions \emph{under
     contexts} as well.
 \end{remark}
 
 \begin{example}\label{ex:complete} The process
   $P=\prf{a}{\prf{b}{\inert}} \paralS
   \prf{\bar{b}}{\prf{\bar{c}}{\inert}}$
   is \emph{not complete} since, \eg $\cinn{a,P}$ but \emph{not}
   $\coutt{a,P}$.  The process is also locked, but can be unlocked by
   the catalyser $[-]\paral\coprf{a}{\prf{c}{\inert}}$ without
   violating channel linearity \ie
   $[P]\paral\coprf{a}{\prf{c}{\inert}} \reduc^\ast \inert$.  The
   inert process, $\inert$ is clearly lock-free and complete.
   \begin{align*}
     % \label{eq:17}
      P_1 & = \prf{a}{\prf{b}{\inert}} \paralS \prf{\bar{b}}{\prf{\bar{c}}{\inert}} \paralS \prf{c}{\prf{\bar{a}}{\inert}} 
     \qquad\qquad\qquad
     & P_2  & = \prf{d}{(\prf{a}{\prf{b}{\inert}} \paral \prf{\bar{b}}{\prf{\bar{c}}{\inert}})} \;\paralS\; \prf{\bar{d}}{\prf{c}{\prf{\bar{a}}{\inert}}}\\
     P_3  & = \prf{a}{\prf{\bar{a}}{\inert}}   &
     P_4  & = \prf{a}{(\prf{b}{\prf{\bar{a}}{\inert}} \paralS \prf{\bar{b}}{\inert})}
   \end{align*}
   By contrast, processes $P_1, P_2, P_3$ and $P_4$ (above) are both
   complete but \emph{not} lock-free.  Note
   that % our discussion rules
   we rule out complete processes such as
   $\prf{a}{\prf{\bar a}{\prf{a}{\inert}}}$ since they violate
   linearity and are thus % would therefore
   not typeable (see \S~\ref{sec:language}).  \exqed
 \end{example}

Our work targets the process class $\CMP \setminus \LF$.  In what
follows we provide a characterisation for this class that is easier
to work with.

\begin{definition}[Deadlock]\label{def:dl}
  \begin{math}
   \dl{P}  \;\deftxt\; \bigl(\not\exists Q \cdot P \reduc Q\bigr) \quad\text{and}\quad P \not\steq \inert 
  \end{math}
\end{definition}

\begin{definition}[Top-Complete]\label{def:top-complete} 
% We define 
% $\TCMP \deftxt \sset{P\in\GProc \;|\; \tcomplete{P}}$    where:
 \begin{align*}
     \tcomplete{P} & \;\deftxt\; (\inn{a,P} \text{ implies } \coutt{a,P}) \;\text{and}\; (\outt{a,P} \text{ implies } \cinn{a,P})  
    \end{align*}
 \end{definition}

\begin{definition}[Potentially Self-Locking]\label{def:psl}
% We define
% the class
$\SL \deftxt \sset{P\in\CMP \;|\; 
%\complete{P} \text{ and } 
\pslfl{P}}$   
 where:
  \begin{align*}
    \slfdl{P} & \;\deftxt\; \dl{P} \text{ and } \tcomplete{P} \\
    \pslfl{P}  &\;\deftxt\; \exists \ectx{-},Q \cdot \bigl(P\reducAst \ectx{Q}  \text{ and } \slfdl{Q} \bigr)
  \end{align*}
\end{definition}

A self-deadlocked processes, \slfdl{P}, denotes a deadlocked process
that cannot be unlocked by a context without violating the linearity
discipline, since the \resp actions are already present in the
process, \ie \tcomplete{P}.  This, together with \dl{P}, also
guarantees that these \resp actions will \emph{never} be released.  A
potentially self-locking process, \pslfl{P}, contains an execution
that leads to a top-level sub-process, \ie $Q$ in \ectx{Q}, that is
self-deadlocked, \slfdl{Q}: \tcomplete{Q} then guarantees $Q$ cannot
interact with any of the future reductions of \ectx{-}.

\begin{example}
  \label{ex:psl} Recall the processes in \exref{ex:complete}. Process
  $P_1$ is self-locking, $\slfdl{P_1}$, and thus potentially
  self-locking as well, \pslfl{P_1}.  Although $P_2$ is \emph{not}
  self-locking, $\neg\slfdl{P_2}$ --- it is not deadlocked and can
  reduce by interacting on name $d$ --- it is \emph{potentially}
  self-locking, \pslfl{P_2}, since $P_2 \reduc P_1$.

  Both $P_3$ and $P_4$ are self-locking as well, but constitute
  instances of the self-holding deadlocked processes studied in
  \cite{GiuntiR13}: in both cases, the locked resource (port $a$) is
  blocked by the co-action under the prefix of the same process.  Such
  locks may be detected by a local analysis of the process prefixed by
  the blocked action.  By contrast, in order to determine \pslfl{P_1}
  and \pslfl{P_2}, the analysis needs to spread across parallel
  processes. \exqed
\end{example}

The main result of the section is that \SL\ characterises $\CMP \setminus \LF$.

\begin{theorem}
  \label{thm:PSL-eq-CMP-less-LF}
  \begin{math}
    \SL = \CMP \setminus \LF
  \end{math}
\end{theorem}
\begin{proof}
  See \secref{sec:proofs}.
\end{proof}

\section{Static Detection for Potentially Self-Locking Processes} % \af{2pg}
\label{sec:stat-detect-potent}
\begin{figure}[t]
  \textbf{Environment Operations}
  \begin{mathpar}
     \inference{}{\env + \emptyset = \env} \quad
    \inference{\env_1+\env_2 = \env_3 \quad a \not\in\dom{\env_2}}{(\env_1,a\!:\!\prV) + \env_2 = \env_3, a\!:\!\prV} \quad
    \inference{\env_1+\env_2 = \env_3}{\env_1,a\!:\!\prV + \env_2,a\!:\!\prVV = \env_3, a\!:\!\prV+\prVV} % \\\\\\
    % \inference{}{\env \oplus \emptyset = \env} \quad
    % \inference{\env_1\oplus\env_2 = \env_3 \quad a \not\in\dom{\env_2} \quad \prV\neq\mio}{\env_1,a\!:\!\prV \oplus \env_2 = \env_3, a\!:\!\prV} \quad
    % \inference{\env_1+\env_2 = \env_3 \quad \prV+\prVV = \mio}{\env_1,a\!:\!\prV \oplus \env_2,a\!:\!\prVV = \env_3}
    
  \end{mathpar}
\textbf{Layered Environments and Verdicts, and Operations}
  \begin{align*}
    \envv \in\LEnv&\bnfdef \epsilon \bnfsepp \env;\envv & \verV \in\Verd&\bnfdef \envv \bnfsepp \verdl 
  \end{align*}
    % \textbf{Layered Environments and Verdict Operations}
    \begin{mathpar}
      \inference{}{\envv + \epsilon = \envv} \quad
      \inference{\envv_1 + \envv_2 = \envv_3}{\env_1;\envv_1 + \env_2;\envv_2 = (\env_1+\env_2);\envv_3}  \quad
      \inference{}{|\epsilon| = \emptyset} \quad 
      \inference{|\envv| = \env'}{|\env;\envv| = \env + \env'}
    \end{mathpar}
    \begin{align*}
      \verPrf{\env}{\verV} & \deftxt
                             \begin{cases}
                               \verdl% \verV
                               & \text{ if } \verV=\verdl  \text{ or } % \env = \emptyset
                               % \\
                               % \verdl & \text{ if }
                               \bigl(\verV=\envv  \text{ and }
                               \dl{\env} \text{ and }\overline{\env} \subseteq \flatt{\envv}\bigr)\\
                               \env;\verV & \text{ otherwise}
                             \end{cases}\\
      \verMrg{\verV_1}{\verV_2} & \deftxt
                               \begin{cases}
                                 \verdl & \text{ if } \verV_1=\verdl \text{ or } \verV_2=\verdl\\
                                 % \verV_2 & \text{ if } \verV_1=\epsilon\\
                                 % \verV_1 & \text{ if } \verV_2=\epsilon\\
                                  \verdl   & \text{ if }
                                  % \begin{cases}
                                    \verV_1=\env_1;\!\envv_1, % i\in1..2
                                     \verV_2=\env_2;\!\envv_2,
                                  % \text{ and } % \\
                                  \dl{\env_1 \!+\! \env_2} \text{ and }\overline{\env_1 \!+\! \env_2} \subseteq \flatt{\envv_1\!+\!\envv_2} % + \flatt{\envv_2}
                                  % \end{cases}
\\
                                   % \verPrf{(\env_1\oplus\env_2)}{
                                     \verMrg{\envv_1}{\envv_2}
                                     % }
                                   & \text{ if }
                                   % \begin{cases}
                                     \verV_1=\env_1;\envv_1, % i\in1..2
                                     \verV_2=\env_2;\envv_2
                                     \text{ and } \complete{\env_1 +  \env_2}
                                     % \\
                                   %   \overline{\env_1 +  \env_2} \cap \flatt{\envv_1} +\flatt{\envv_2} = \emptyset
                                   % \end{cases}
\\
                                   \verV_1 + \verV_2 & \text{ otherwise (since } \verV_1=\envv_1, \verV_2=\envv_2\text{)} 
                               \end{cases}
    \end{align*}
    \textbf{Compositional Static Analysis Rules}
     \begin{mathpar}
  \inference[\rtit{dNil}]{}{\pslseq{\inert}{\emptyset} }  % \and
  \quad
  \inference[\rtit{dIn}]{\pslseq{P}{\verV}}{\pslseq{\prf{a}{P}}{(\verPrf{a:\mi}{\!\verV})}}
  % \\\\\\
  \quad
  \inference[\rtit{dOut}]{\pslseq{P}{\verV}}{\pslseq{\prf{\bar{a}}{P}}{(\verPrf{a:\mo}{\!\verV})}}  % \and
  \quad
  \inference[\rtit{dPar}]{\pslseq{P_1}{\verV_1} & \pslseq{P_2}{\verV_2}  }{\pslseq{P_1\paral P_2}{\verV_1 \oplus \verV_2}} 
   \end{mathpar}
  \caption{Static Analysis for Potential Self-Deadlock}
  \label{fig:static-analysis}
  % \hrulefill
\end{figure}

We devise an algorithm for detecting potentially self-locking
processes. To be scalable, the algorithm is \emph{compositional}.
% whereby the analysis of composite processes, \eg
% $P_1 \paral P_2$, is defined in terms of the individual analyses of
% its constituents, \ie $P_1$ and $P_2$.
The intuition behind %our algorithm
is that of constructing \emph{layers} of permission environments
$\env_1;\ldots;\env_n$, \emph{approximating} the prefixing found in the
process being analysed, and then checking whether this structure
satisfies the two conditions defining self-deadlock (see \slfdl{-} in
\defref{def:psl}), namely that the top environment $\env_1$ represents
a \emph{deadlock} and that the layered structure is, in some sense,
\emph{top-complete}.

\begin{example}\label{ex:psl-alg-overview} We % want to
  determine that the process $P_1$ from \exref{ex:complete} is
  (potentially) self-locking
  % \begin{equation}
  %   \label{eq:6}
  %   \prf{a}{\prf{b}{\inert}} \paralS \prf{\bar{b}}{\prf{\bar{c}}{\inert}} \paralS \prf{c}{\prf{\bar{a}}{\inert}}
  % \end{equation}
  by constructing the list of layered environments 
  \begin{equation*}
    \underbrace{(a:\mi,b:\mo,c:\mi)}_{\env_1} \; ;\; \underbrace{(a:\mo,b:\mi,c:\mo)}_{\env_2}\; ; \; \epsilon
  \end{equation*}
  and checking that
  \begin{itemize}
  \item the top environment $\env_1$ does not contain any matching
    permissions, \ie $\mio\not\in\cod{\env_1}$ --- this implies that
    the (composite) process is deadlocked;
  \item that \emph{all} the \resp dual permissions are in $\env_2$ ---
    this implies that the (composite) process is blocking itself and
    cannot be unblocked by an external process composed in parallel
    with it.
  \end{itemize}
  The main challenge of our compositional analysis is to detect
  eventual self-deadlocks in cases when the constituent are dispersed
  across a number of parallel processes.
 % \begin{equation}
 %    \label{eq:14}
 %    \prf{d}{(\prf{a}{\prf{b}{\inert}} \paral \prf{\bar{b}}{\prf{\bar{c}}{\inert}})} \;\paralS\; \prf{\bar{d}}{\prf{c}{\prf{\bar{a}}{\inert}}}
 %  \end{equation}
  In the case of $P_2$ of \exref{ex:complete}, we need to analyse the
  parallel (sub) processes
  $\prf{d}{(\prf{a}{\prf{b}{\inert}} \paral
    \prf{\bar{b}}{\prf{\bar{c}}{\inert}})}$
  and \prf{\bar{d}}{\prf{c}{\prf{\bar{a}}{\inert}}} \emph{in
    isolation}, and then determine the eventual deadlock once we merge
  the sub-analyses; recall, from \exref{ex:psl}, that
  $P_2$ % \eqref{eq:14}
  reduces to $P_1$ % \eqref{eq:6}
  from the respective continuations prefixed by $d$ and
  $\bar{d}$. \exqed
\end{example}

Formally, permissions, $\prV,\prVV\in\sset{\mi,\mo,\mio}$, denote
\resp input, output and input-output capabilities. The merge,
$\prV+\prVV$, and complement, $\overline{\prV}$, (partial) operations
are defined as:
\begin{equation*}
  \mi + \mo \deftxt \mio \qquad \qquad\qquad\qquad\overline{\mi}\deftxt\mo\qquad\overline{\mo}\deftxt\mi\qquad \overline{\mio}\deftxt\mio
\end{equation*}
Environments, \env, are partial maps from names to
permissions.  % , $\env \in (\Names \rightharpoonup \sset{\mi,\mo,\mio})$;
We assume the following overloaded notation: complementation,
$\overline{\env}$, % simply
inverts the respective permissions in $\cod{\env}$ whereas deadlock
and complete predicates are defined as:
\begin{equation*}
  \dl{\env} \deftxt \cod{\env} = \sset{\mi,\mo} \qquad \qquad\complete{\env} \deftxt \cod{\env} = \sset{\mio} 
\end{equation*}
The rules in \figref{fig:static-analysis} define the merge operation
over environments, $\env_1+\env_2$ (we elide symmetric rules).
Layered environments, \envv, are lists of environments.  Our static
analysis sequents take the form ${\pslseq{P}{\verV}}$ where \verV\ is
a \emph{verdict}: it can either be a layered environment or \verdl,
denoting a detection.  Layered environments may be merged,
$\envv_1+\envv_2$, or flattened into a single environment, $|\envv|$;
see \figref{fig:static-analysis}.

The static analysis rules are given in \figref{fig:static-analysis},
and rely on two verdict operations.  Prefixing,
$\verPrf{\env\!}{\!\verV}$, collapses to the definite verdict \verdl\
if \verV\ was definite or else \env is deadlocked, \dl{\env}, and
top-complete, % \wrt \verV=\envv,
$\overline{\env} \subseteq \flatt{\envv}$, but creates an
\emph{extended} layered environment
otherwise. % (for optimisation purposes, it does not extend \envv when $\env=\emptyset$);
Verdict merging, $\verMrg{\verV_1}{\verV_2}$, collapses to \verdl\ if
either subverdict is a definite detection, or the combined top
environments, $\env_1 + \env_2$, satisfy environment deadlock and
top-completeness; if $\env_1 + \env_2$ is complete, then it is safe to
discard it and check for self-deadlock in the sub-layers (see
\remref{rem:unsound}) otherwise the verdicts (which must both be
layered environments) are % (approximately)
merged.

\begin{example}\label{ex:psl-alg-deriv} Recall process $P_2$ from \exref{ex:psl-alg-overview}.  We can derive the sequents:
  \begin{align} 
    \label{eq:3}
    \prf{d}{(\prf{a}{\prf{b}{\inert}} \paral \prf{\bar{b}}{\prf{\bar{c}}{\inert}})} & \;\; \triangleright \;(\;d:\mi)\,;\, (a:\mi,b:\mo)\,;\, (b:\mi,c:\mi)\,;\, \epsilon \\ % &
    % \qquad
    \label{eq:4}
    \prf{\bar{d}}{\prf{c}{\prf{\bar{a}}{\inert}}}& \;\;\triangleright\;\;(d:\mo)\,;\, (c:\mi)\,;\, (a:\mo)\,;\,\epsilon
  \end{align}
  For instance, in the case of \eqref{eq:3}, we first derive the
  judgements
  $\pslseq{\prf{a}{\prf{b}{\inert}}\;}{\;a:\mi; b:\mi; \epsilon}$ and
  $\pslseq{\prf{\bar{b}}{\prf{\bar{c}}{\inert}}\;}{\;b:\mo;c:\mo;\epsilon}$
  using rules \rtit{dNil}, \rtit{dIn} and \rtit{dOut}.  Applying
  \rtit{dPar} on these two judgements requires us to calculate
  \begin{equation*}
    \verMrg{(a:\mi; b:\mi; \epsilon)}{(b:\mo;c:\mo;\epsilon)} \; = \; {(a:\mi; b:\mi; \epsilon)} + {(b:\mo;c:\mo;\epsilon)} \;=\; (a:\mi,b:\mo); (b:\mi,c:\mi); \epsilon
  \end{equation*}
  using the definition of \verMrg{\verV_1}{\verV_2} from
  \figref{fig:static-analysis}.  We thus obtain \eqref{eq:3} by
  applying \rtit{dIn} on the resultant judgement.

  Importantly, when we use rule \rtit{dPar} again, this time to merge
  judgements \eqref{eq:3} and \eqref{eq:4}, the definition of
  \verMrg{\verV_1}{\verV_2} allows us to reexamine the environments in
  the sub-layers, since the merged top-layer is complete,
  $\complete{d\!:\mi\!+ d\!:\mo}$, from which we infer that the top
  actions guarding the merged parallel processes will safely interact
  and release the processes in the sub-layers.  Stated otherwise, we
  obtain:
  \begin{equation*}
    \verMrg{\bigl(d:\mi; (a:\mi,b:\mo); (b:\mi,c:\mi);\epsilon\bigr)}{\bigl(d:\mo; c:\mi; a:\mo;\epsilon\bigr)} \;=\; \verMrg{\bigl((a:\mi,b:\mo); (b:\mi,c:\mi);\epsilon\bigr)}{\bigl(c:\mi; a:\mo;\epsilon\bigr)} \;=\; \verdl
  \end{equation*}
  since $\dl{(a:\mi,b:\mo) + (c:\mi)}$ and
  $\;\overline{a:\mi,b:\mo,c:\mi} \subseteq
  \flatt{\bigl((b:\mi,c:\mi);\epsilon\bigr)+
    \bigl(a:\mo;\epsilon\bigr) } = (b:\mi,c:\mi,a:\mo)$. \exqed
\end{example}

\begin{remark} \label{rem:unsound} When merging verdicts, it is unsafe
  to ignore individual complete mappings, \eg $a:\mio$, even though
  this makes the analysis imprecise. It is only safe to ignore them
  (and check for potential deadlocks in lower layers) when the
  \emph{entire} top environment is complete, \ie \complete{\env}.  As
  a counter-example justifying this, % approximation,
  consider the lock-free process
  \begin{math}
    (\prf{a}{\coprf{b}{\inert}} \paral \prf{b}{\inert}) \paral \coprf{a}{\inert}
  \end{math}.
   We currently deduce
   \begin{align*}
     \pslseq{(\prf{a}{\coprf{b}{\inert}} \paral \prf{b}{\inert})\;}{\;(a\!:\mi,b\!:\mi);b\!:\mo;\epsilon} \;\quad\text{and}\quad\;
     \pslseq{\coprf{a}{\inert}\;}{\;a\!:\mo;\epsilon} \;\quad 
     \text{ where }\;\quad \overline{a\!:\mi,b\!:\mi + a\!:\mo} = a\!:\mio,b\!:\mo \not\subseteq b\!:\mo
   \end{align*}
   However, eliding $a\!:\mio$ from the analysis, \ie assuming that
   $(a\!:\mi,b\!:\mi) + a\!:\mo \;=\; b\!:\mo$, yields an unsound
   detection.  Precisely, when merging the sub-verdicts for rule
   \rtit{dPar},
   \verMrg{\bigl((a\!:\mi,b\!:\mi);b\!:\mo;\epsilon\bigr)}{\bigl(a\!:\mo;\epsilon\bigr)},
   we would first obtain $\dl{(a\!:\mi,b\!:\mi) + a\!:\mo}$ and
   moreover that $\overline{(a\!:\mi,b\!:\mi) + a\!:\mo} = b\!:\mo$ is
   a subset of ${\flatt{(b\!:\mo;\epsilon) + \epsilon}} = b\!:\mo$,
   which yields \verdl according to
   \figref{fig:static-analysis}. \exqed
\end{remark}
\smallskip
\noindent
We expect the judgement \pslseq{P}{\verdl} to imply \pslfl{P}, which
would in turn imply $\neg\lf{P}$ by
\thmref{thm:PSL-eq-CMP-less-LF}. We leave the proof of the first
implication for future work.

% \begin{conjecture}
%   \pslseq{P}{\verdl} implies \pslfl{P}
% \end{conjecture}

% \input{detect-psl-old.tex}

%%% Local Variables:
%%% mode: latex
%%% TeX-master: "position-paper"
%%% End:

\section{Disentangling Potentially Self-Locking Processes} %\af{1pg}
\label{sec:disent-potent-self}
\begin{figure}[t]
  % \textbf{Process Disentangling}
  \begin{align*}
    \dentO{\env\!}{\!\inert} &\deftxt \inert & 
     \dentO{\env\!}{\!P\paral Q} &\deftxt (\dentO{\env\!}{\!P}) \paral (\dentO{\env\!}{\!Q})\\
    \dentO{\env\!}{\!\coprf{a}{P}} & \deftxt
                                 \begin{cases}
                                   \coprf{a}{\inert}\paral P &   \env(a) = \mo\\
                                   \coprf{a}{(\dentO{\env\!}{\!P})}  & \text{otherwise}
                                 \end{cases}  &
    \dentO{\env\!}{\!\prf{a}{P}} & \deftxt
                                 \begin{cases}
                                   \prf{a}{\inert}\paral P &   \env(a) = \mi\\
                                   \prf{a}{(\dentO{\env\!}{\!P})}  & \text{otherwise}
                                 \end{cases}\\[0.5em]
    \dentT{\env\!}{\!\inert} &\deftxt \inert & 
     \dentT{\env\!}{\!P\paral Q} &\deftxt (\dentT{\env\!}{\!P}) \paral (\dentT{\env\!}{\!Q})\\
    \dentT{\env\!}{\!\coprf{a}{P}} & \deftxt
                                 \begin{cases}
                                   \coprf{a}{\inert}\paral (\dentT{\env\!}{\!P}) &   \env(a) = \mo\\
                                   (\dentT{\env\!}{\!P}) &   \env(a) = \mi\\
                                   \coprf{a}{(\dentT{\env\!}{\!P})}  & \text{otherwise}
                                 \end{cases}  &
    \dentT{\env\!}{\!\prf{a}{P}} & \deftxt
                                 \begin{cases}
                                   \prf{a}{(\dentT{\env\!}{\!P})} \paral \coprf{a}{\inert} &   \env(a) = \mi\\
                                   \prf{a}{(\dentT{\env\!}{\!P})}  & \text{otherwise}
                                 \end{cases}
  \end{align*}
  \caption{Disentangling for Potential Self-Deadlock}
  \label{fig:disentangling-translation}
  % \hrulefill
\end{figure}

To illustrate the ultimate aim of our study, we outline
possible disentangling functions that refactor a potentially
self-deadlocked process into a corresponding lock-free process.  These
disentangling functions are meant to be used in conjunction with the
detection algorithm of \secref{sec:stat-detect-potent} as a static
analysis tool for automating the disentangling of processes.
There are a number of requirements that a disentangling algorithm
should satisfy.  For instance, it should not violate any safety
property that is already satisfied by the entangled process (\eg if an
entangled process $P$ type-checked according to some typing discipline,
\ie \tproc{\env}{P} , the resulting disentangled processes, say $Q$,
should still typecheck \wrt the same type discipline/environment, \ie
\tproc{\env}{Q}).  Additionally, one would also expect the resultant
disentangled process to be lock-free, as expressed in \defref{def:lf},
or at the very least to resolve a subset of the locks detected.  But
there are also a number of additional % usually
possibilities for what constitutes a valid process disentangling.
Within the simple language of \secref{sec:language}, we can already
identify at least two (potentially conflicting) criteria:
\begin{enumerate}
\item the order of name usage respects that dictated by the innermost
  prefixing of the entangled process.  Stated otherwise, any locks are
  assumed to be caused by prefixing at the top-level of the process.
\item the order of input prefixes in the entangled process should be preserved.
\end{enumerate}
We envisage a straightforward extension to the system
\pslseq{P}{\verV} of \secref{sec:stat-detect-potent}, with extended
detection reports, $\langle\verdl,\env\rangle$. The tuple
$\langle\verdl,\env\rangle$, in some sense, \emph{explains} the source
of the problem detected by including the offending top-layer
environment of a self-deadlock, \env; this information is then used by
the disentangling procedure to refactor the detected process.

\figref{fig:disentangling-translation} defines two disentangling
functions that take this (top-layer) environment and the \resp
detected process as input, and return a refactored process as output.
The first function, \dentO{\env\!}{\!P}, translates problematic
prefixing (as dictated by \env) into parallel compositions. The second
function \dentT{\env\!}{\!P} operates asymmetrically on input and
output prefixes: whereas problematic outputs are treated as before,
blocked inputs are \emph{not} parallelised; instead the \resp output
is \emph{pulled out} at input level.

\begin{example}
  \label{ex:disentag}
  The algorithm of \secref{sec:stat-detect-potent} could detect $P_5$
  (below) as \pslseq{P_5\:}{\;\langle\verdl,(a:\mi,c:\mo)\rangle},
  where $\env = (a:\mi,c:\mo)$
  % The % deadlocked
  % process $P_5$ (below) % \eqref{eq:16}
  % is detected by .
  \begin{equation*}
    % \label{eq:16}
    P_5=\prf{a}{\coprf{b}{\prf{c}{\inert}}}\;\paralS\;\coprf{c}{\prf{b}{\coprf{a}{\inert}}}
  \end{equation*}
  Using the offending top-layer environment $\env$, we can apply the
  two disentangling algorithms of
  \figref{fig:disentangling-translation} and obtain the following:
  \begin{align}
    \label{eq:15}
    \dentO{\env\!}{\!P_5} & = (\prf{a}{\inert} \paral \coprf{b}{\prf{c}{\inert}})\paral(\coprf{c}{\inert}\paral\prf{b}{\coprf{a}{\inert}}) \\
    \label{eq:18}
    \quad\dentT{\env\!}{\!P_5} & = (\coprf{a}{\inert} \paral \prf{a}{\coprf{b}{\prf{c}{\inert}}})\paral(\coprf{c}{\inert}\paral\prf{b}{\inert}) 
  \end{align}
  While both refactored processes are lock-free, it turns out that the
  first disentangling function observes the first criteria: in the
  refactored process,~\eqref{eq:15}, interactions on $a$ and $c$
  happen \emph{after} interactions on $b$, since these names are
  (both) prefixed by $b$ (and $\bar{b}$) at the innermost level of
  $P_5$.  Conversely, the second disentangling function observes the
  second criteria discussed above: in the refactored process,
  \eqref{eq:18}, the input prefixing that orders $a$ before $c$ in
  $P_5$ is preserved (this was not the case in \eqref{eq:15}).  Note
  that both refactorings preserve channel linearity (a safety
  criteria) while returning lock-free processes.\exqed
\end{example}

% \begin{figure}[t]
%   \textbf{Process Disentangling}
%   \begin{mathpar}
%     \disseq{P}{Q,\envv}
%   \end{mathpar}
%   \caption{Disentangling Interleaved with detection}
%   \label{fig:disentangling-detection}
%   % \hrulefill
% \end{figure}

% \input{disentang-old.tex}

%%% Local Variables:
%%% mode: latex
%%% TeX-master: "position-paper"
%%% End:

\section{Conclusion} % \af{0.5pg}
\label{sec:conclusion}
We have outlined our strategy for automating correct disentangling of
locked
processes, % in terms of a lock detection algorithm and possible disentangling procedures,
generalising preliminary results previously presented~\cite{GiuntiR13}.
Although we limited our discussion to a very simple language --- the
variant of the finite \ccs\ without recursion, choice or name scoping
--- this was expressive enough to focus on the usefulness of the
concepts and techniques we propose, \ie resolving circular locks
across parallel compositions. We define precisely the class of
(dead)locked processes within this setting, and provide a faithful
characterisation of them in terms of a novel notion: potentially
self-locking processes.%, \defref{def:psl}.
We also devised a compositional algorithm to statically detect these
processes and unlock them, improving previous results
(cf.~\cite{GiuntiR13}).
In particular, Giunti and Ravara~\cite{GiuntiR13} used a different
technique %to ours
(based on balanced session types) and could only disentangle
self-holding deadlocks such as those in processes $P_3$ and $P_4$ of
\exref{ex:complete}. %However, they
The technique does not support reasoning about (and disentangle) locks
across parallel compositions, such as those shown for processes $P_1$
and $P_2$ of \exref{ex:complete} and $P_5$
of % \exref{ex:psl-alg-overview}
\exref{ex:disentag}.

% Although the exposition was carried out for a simple language,

We expect the concepts and techniques developed to carry
over %, at least in part,
to more expressive languages. We %have performed exploratory investigations for
are considering language extensions such as process recursion,
unrestricted channel names (to allow non-determinism), and
value-passing. For instance, disentangling the value-passing
program (an extension of the process $P_5$ in \exref{ex:disentag})
%\begin{equation*}% \label{eq:17}
$
        P_6=\prf{a(x)}{\prf{\bar{b}\langle
            x+1\rangle}{\prf{c(y)}{\inert}}}
        \;\paralS\;
        \prf{\bar{c}\langle 5\rangle}{\prf{b(z)}{\prf{\bar{a}{\langle7\rangle}}{\inert}}}
$
%\end{equation*}
may not be possible for certain disentangling functions (and criteria)
\eg $\dentO{\env\!}{\!P}$, whereas others may require auxiliary
machinery, \eg the $\textbf{findVal}(-)$ function used by
\cite{GiuntiR13} for pulling out the \resp output values in
$\dentT{\env\!}{\!P}$; in (complete) linear settings, there is a
unique output for any particular channel, which can be obtained
through a linear scan of the process. The input binding structure may
also make certain processes impossible to disentangle. \Eg consider a
modification in $P_6$ above % (\ref{eq:17})
where the value $7$ is changed to the bound value $z$. This would
create a circular binding dependency: one between the input on channel
$a$ and the output on $b$ through variable $x$, but also another one
between the input on $b$ and the output on $a$ through variable $z$.
These issues %, along with many others,
will all be considered in future work.

%%% Local Variables:
%%% mode: latex
%%% TeX-master: "position-paper"
%%% End:

\section*{Acknowledgements}
Adrian Francalanza was supported by the grants
ECOST-STSM-IC1201-280114-038254
and
ECOST-STSM-IC1201-250115-054509.
Marco Giunti was supported by the grant
ECOST-STSM-IC1201-220713-032903 and by the Software Testing Center,
Centro de Neg\'ocios e Servi\c cos Partilhados do Fund\~ao.
Ant\'onio Ravara was supported by the grant
ECOST-STSM-IC1201-210713-033367.
Marco Giunti and Ant\'onio Ravara were also supported by the
grant FCT/MEC NOVA LINCS PEst UID/CEC/04516/2013.

\bibliography{sessions}  
\bibliographystyle{eptcs} 

\newpage
\appendix

% \section{Examples}
% \label{sec:examples}
% \input{appdx.tex}

% \section{Disentangling Rules}
% \label{sec:disentangling-rules}
% \input{appdx-disentangling.tex}
\section{Proofs}
\label{sec:proofs}

This section is devoted to the proof of Theorem~\ref{thm:PSL-eq-CMP-less-LF}. We start with some 
auxiliary definitions and lemmas.
\smallskip

Given a  process~$P$ of Figure~\ref{fig:language}, we indicate with $\Names(P)$ the subset of $
\Names$ induced by the rule
$\Names(a.P) = \{a\}\cup\Names(P)$: the remaining cases are homomorphic. 
We use $\sqcup$ for disjoint union of sets.

We remember that we assume that the processes~$P$ of our interest are linear, that is they never 
contain two or more inputs or  outputs on the same channel, and deploy the following results.

%\CMP  \deftxt \sset{P \;|\; \complete{P}}$    where:
%\begin{align*} 
%      & \quad\complete{P}  \;\deftxt\; \cinn{a,P} \text{ iff } \coutt{a,P}  \qquad \text{ and }\\
%    \cinn{a,P} & \;\deftxt\; P \steq \ctx{Q} \text{ and } \inn{a,Q}      
%   \quad\qquad\coutt{a,P}  \;\deftxt\; P \steq \ctx{Q} \text{ and } \outt{a,Q} 
%    \end{align*}

\begin{lemma}\label{lem:sync}
If $P\rightarrow P'$ then there is port, $a$, such that 
$\Names(P')\sqcup\{a\} =\Names(P)$ and  $ \snc{a,P}$.
\end{lemma}
\begin{proof}
By induction on the rules of Figure~1; straightforward.
\end{proof}

\begin{corollary}\label{cor:sync}
If $P\rightarrow^* Q$ then the following  holds:
\begin{enumerate}
\item $\Names(Q)\subseteq\Names(P)$ 
\item if $\Names(P)\backslash\Names(Q) =\{a,\dots\}$ then there exists a $P_a$ and a $P'_a$ such that
$P\rightarrow^*P_a \rightarrow P'_a\rightarrow^*Q$ with 
$\Names(P'_a)\sqcup\{a\} =\Names(P_a)$ and  $ \snc{a,P_a}$. 
\end{enumerate}
\end{corollary} 
 
\begin{lemma}\label{lem:cin-persistent}
If $\cinn{a,P}$ ($\coutt{a,P}$) and $P\rightarrow^* Q$ then exactly one of the following cases holds:
\begin{enumerate}
\item $\cinn{a,Q}$ ($\coutt{a,Q}$) 
\item $a\not\in\Names(Q)$ and  there exists an $R_a$ and an $R'_a$ such that $P\rightarrow^* R_a\rightarrow R'_a\rightarrow^* Q$ and
$\Names(R'_a)\sqcup\{a\} =\Names(R_a)$ and  $ \snc{a,R_a}$. 
\end{enumerate}
\end{lemma}
\begin{proof}
We have two cases corresponding to (i) $a\in\Names(Q)$ or   (ii) $a\not\in\Names(Q)$.
In case (i), assume $P\rightarrow  P_1\rightarrow_1 \cdots \rightarrow_n Q_n\rightarrow Q$. We 
proceed by induction on~$n$.  From Lemma~\ref{lem:sync} we know that there exists $b$ such that 
$\Names(P_1)\sqcup\{b\} =\Names(P)$ and  $ \snc{b,P}$. Since $\Names(Q)\subset\Names(P_1)
$, we infer $a\in\Names(P_1)$ and in turn $a\ne b$. From this and $\cinn{a,P}$ we deduce that
$\cinn{a,P'}$. Now assume that $\cinn{a,Q_n}$. From Lemma~1 we deduce   
$\Names(Q)\sqcup\{c\} =\Names(Q_n)$ for some~$c\ne a$: thus $\cinn{a,Q}$. 
The case  $\coutt{a,P}$  is analogous.
Case (ii) is a direct consequence of Corollary~\ref{cor:sync}.

\end{proof}
    
\begin{lemma}\label{lem:cmp-persistent}
If $P\in \CMP$ and $P\rightarrow^* Q$ then  $Q\in\CMP$.
\end{lemma}

\begin{lemma}\label{lem:finite-reduction}
For any~$P$ there exists~$Q$ such that $P\rightarrow^* Q$ and $Q\not\rightarrow$.
\end{lemma}

%$\SL \deftxt \sset{P\in\CMP \;|\; \complete{P} \text{ and } 
%\pslfl{P}}$   where:
%  \begin{align*}
% \slfdl{P} & \;\deftxt\; \dl{P} \text{ and } \tcomplete{P} \\
%   \pslfl{P}  &\;\deftxt\; \text{Exists }\ectx{Q} \text{ such that } P\reducAst \ectx{Q}  \text{ and }%\slfdl{Q} 
 %\end{align*}

\medskip\noindent{\bf Proof of Theorem~\ref{thm:PSL-eq-CMP-less-LF}}.
To show the right to the left direction, assume $P\in \CMP \setminus \LF$. By definition:
\[
\CMP\setminus LF  \deftxt
\{P\in \CMP \ |  \ \exists (Q,a)\ . \  P\rightarrow^* Q \ \land \ \wt{a,Q} \Rightarrow\forall R\  . \  Q\rightarrow^* R  
\Rightarrow \neg\snc{a,R}\}
\] 
Let $Q_a$ be a distinctive redex of $P$: thus $\inn{a,Q_a} \;\text{ exor }\;\outt{a,Q_a}$.
Assume $\inn{a,Q_a}$ and consider $R_{\textbf{stop}}$ such that 
$Q_a\rightarrow^*R_{\textbf{stop}}\not\rightarrow$, which does exists by Lemma~\ref{lem:finite-reduction}.
By Lemma~\ref{lem:cin-persistent} we know that $\inn{a, R_{\textbf{stop}}}$: from  $\neg\snc{a,R_{\textbf{stop}}}$ we 
infer  $\neg\outt{a,R_{\textbf{stop}}}$.
From Lemma~\ref{lem:cmp-persistent} we infer $R_{\textbf{stop}}\in\CMP$: thus $\coutt{a,R_{\textbf{stop}}}$.
Therefore $\dl{R_{\textbf{stop}}}$ and $\tcomplete{R_{\textbf{stop}}}$, as required.
The case  $\outt{a,Q_a}$ is analogous.

\smallskip
To see the  left to the right direction, assume that $P\in\CMP$ and that 
$P\reducAst \ectx{Q}$ with $\dl{Q} $ and  $\tcomplete{Q}$. Note that this excludes the case $Q
\equiv\inert$: therefore $\Names(Q)\ne\emptyset$,  and in turn $\Names(P)\ne\emptyset$, 
because of Corollary~\ref{cor:sync}. From $\Names(Q)\ne\emptyset$ and the rules of structural congruence we 
infer that there is $a\in\Names(Q)$  such that (i) $Q\equiv a.Q'\paralS Q''$ or (ii) $Q\equiv \overline 
a.Q'\paralS Q''$. In case (i) we infer $\inn{a,Q}$; from $Q\not\rightarrow$ we deduce $\neg
\outt{a,Q}$; in case (ii) we infer $\outt{a,Q}$ ; from $Q\not\rightarrow$ we deduce $\neg\inn{a,Q}$.
In both cases we infer $\wt{a,Q}$, and in turn $\neg\snc{a,Q}$ which completes the proof since  $Q$ 
has no redexes: that is,  $P\in\CMP\setminus LF $.
\qed

\end{document}